\documentclass[journal]{IEEEtran}
\usepackage{amssymb, amsthm, bm}
\usepackage[tbtags]{amsmath}
\usepackage{datetime}
\usepackage{graphicx}
\graphicspath{{figures-pdf/}}
\usepackage[caption=false, font=footnotesize, labelfont=rm, textfont=rm]{subfig}

\usepackage{cite} 
\usepackage{url} 
\usepackage[aboveskip=1pt]{subcaption}
\usepackage{diagbox}
\usepackage{multirow, bigstrut}
\usepackage{arydshln} 
\usepackage[mathlines]{lineno}

\usepackage{verbatim} 
\usepackage{comment} 

\usepackage{hyperref}
\hypersetup{linktocpage=true, pdfborderstyle={/S/S/W 1}, bookmarksopen=true, 
colorlinks, linkcolor=blue, anchorcolor=blue, citecolor=blue}

\newlength\OneImW
\newlength\twofigwidth
\newlength\ThreeImW
\newtheorem{property}{Property}

\newtheorem{corollary}{Corollary}

\newcommand{\Ze}{\mathcal{Z}_{2^e}}

\newcommand{\vvert}[0]{{\, \vert \, }}

\begin{document}

\title{Graph Structure of the Generalized Tent Map over Ring $\mathbb{Z}_{2^e}$}

\author{Kai Tan, Chengqing Li

 \thanks{K. Tan and C. Li are with the School of Computer Science, Xiangtan University, Xiangtan 411105, Hunan, China (\protect\url{DrChengqingLi@gmail.com}).}  
 
 \thanks{This work was supported by the National Natural Science Foundation of China (no.~92267102).}
}    
    \markboth{IEEE Transactions}{Li \MakeLowercase{\textit{et al.}}}
    \IEEEpubid{\begin{minipage}{\textwidth}\ \\[12pt]\centering
    \\[2\baselineskip]
    0018-9448 \copyright 2024 IEEE. Personal use is permitted, but republication/redistribution requires IEEE permission.\\
    See http://www.ieee.org/publications\_standards/publications/rights/index.html for more information.\\ \today{}\space(\currenttime)
    \end{minipage}
    }

\maketitle

\begin{abstract}
The structure of functional graphs of nonlinear systems provides one of the most intuitive methods for analyzing their properties in digital domain. The generalized Tent map is particularly suitable for studying the degradation of dynamic behaviors in the digital domain due to its significantly varied dynamical behaviors across different parameters in the continuous domain. This paper quantifies the generalized Tent map under various parameters to investigate the dynamic degradation that occurs when implemented in fixed-precision arithmetic.  Additionally, the small period problem caused by the indeterminate point in the fixed-point domain is addressed by transforming the parameters. The period of the mapping in the fixed-point domain is optimized by introducing disturbances, resulting in an optimized mapping that is a permutation mapping. Furthermore, the paper reveals the dynamic degradation of analogous one-dimensional maps in fixed-precision domains and provides a more orderly design for this map.
\end{abstract}

\section{Introduction}


\IEEEPARstart{I}{n} numerous studies have revealed the degradation of dynamic properties of chaotic systems in digital devices. When chaotic dynamical systems are implemented in computers or digital domains, their statistical properties differ significantly from those in the continuous domain \cite {Boghosian:Pathology:2019}. 
For the logistic map in the 64-bit floating-point arithmetic domain, the map's error statistics vary with parameters and the impact of these errors on dynamic behavior is revealed \cite{Oteo:Double:2007}. Sequences generated by the Chebyshev map, Logistic map, and Tent map encounter various issues such as short periods and a limited number of unique periods when implemented in a finite-precision digital domain \cite{Li:Dynamic:2019}. Galias examined how rounding errors and the finiteness of the state space affect the properties of nonlinear maps by an efficient graph-based algorithm designed to find basins of attraction \cite{Galias:TCSI2023, Galias:TCSII2021}. To mitigate this degradation, \cite{Wang:TCSI2016} proposed a chaotic system based on a chaos generation strategy controlled by random sequences.
Additionally, \cite{OZTURK2018395} introduced a method that fills the least significant bits (LSB) of each state with random bits, although this method is only applicable to binary shift chaotic maps (BSCM). Subsequently, a high-dimensional Baker's map has been defined as a BSCM for all dimensions \cite{Ozturk:Baker:TCASI2019}. For enhancing performance and resource efficiency while maintaining chaotic behavior and strong pseudorandom properties, a new chaos-based PRNG architecture using the half-unit-biased format is proposed, advancing its applicability in instrumentation and measurement systems \cite{Silva:TIM2023}.

The Tent map, as an elementary chaotic system, is a source of some applications, such as pseudorandom bit generators (PRBGs) \cite{Valle:CSF2022,Addabbo:digital:2006}, optimization algorithms \cite{Hesb:Tent:CASII23}, analog-digital converters \cite{Basu:Algorithmic:2018} and encryption algorithm \cite{Yi:block:2002, Wang:IS2021, Hua:IS2021}. The Tent map is defined by $f_{\mu}(x)=\mu(\min(x, 1-x))$, where $\mu$ is real-valued.
As $\mu$ increases, the dynamic behavior of the tent map changes from predictable to chaotic.
If $\mu=2$, the Tent map becomes chaotic, and it is proved to be neither topologically stable, topologically stable, orbital displacement, nor countable inflate \cite{Carrasco:Stability:2020}.
However, the Tent map's simple structure is far from enough to show the complexity of a one-dimensional map in dynamic analysis.
In \cite{Jessa:period:2002}, a tent-like map, a continuous piecewise linear map, and the relevant Sawtooth map are proposed.

The graph structure of a functional graph, also called the state-mapping network, provides an intuitive approach to analyzing the dynamic properties of chaotic systems in the different digital domains, including the digital domain \cite{Li:Dynamic:2019} and the finite ring  \cite{cqli:Cheby:TIT25}. A high-performance general framework was proposed for generating 2-D chaotic maps, and the dynamic properties of these new chaotic maps in the digital domain were demonstrated through functional graphs \cite{Zhang:TIE2023}. 
The functional graph of Baker's map in the fixed-point domain shows apparent self-similarity \cite{cqli:Baker:TC25}.

Diversity is the essential character of the one-dimensional map in the fixed-precision domain for dynamical analysis.
The generalized Tent map $f: [0, 1]\rightarrow[0, 1]$ is defined by
\begin{equation}\label{eq:define_f}
  f(x)=
  \begin{cases}
    b+\frac{1-b}{a}x  & \mbox{if } 0\leq x<a;\\
    \frac{1-x}{1-a}   & \mbox{if } a\leq x<1.
  \end{cases}
\end{equation}
The different dynamic properties of this map in the continuous domain with various parameters are analyzed in \cite{Susan:Dynamics:1998}.
William studies this map with inverse limits and presents a sequence of theorems \cite{William:Interesting:2014}.

\IEEEpubidadjcol

This paper studies the properties of functional graphs generated by iterating the generalized Tent map in the digital domain. These include the disturbance of a state with the precision increase and the in-degree distribution.
Then, the perturbations at the fixed and periodic points are analyzed.
Moreover, other properties are revealed at the same time.
When the generalized Tent map~\eqref{eq:define_f} is non-chaotic, its dynamic properties degrade slightly and become smaller with higher precision in the digital domain.
And when the generalized Tent map~\eqref{eq:define_f} is chaotic, its dynamic properties become erratic in the digital domain.
The reason, in short, is that Devaney's definition of chaos depends on the infinite subset.
The definition contradicts the digital domain.
The topologically transitive means that the interval $I$ continuously extends with the iteration.
So, the dynamical perturbations are more severe under topologically transitive.
The property that periodic points are dense disappears because the state is finite.
In conclusion, chaos also disappears in the digital domain.

The rest of this paper is organized as follows. Section~\ref{sec:SMN} provides a detailed analysis and mathematical proof of some properties of the functional graph of the generated tent map in the fixed-point domain. Section~\ref{sec:Optimization} discusses the process of obtaining a mapping with improved dynamic properties by transforming the mapping parameters and introducing disturbances, and it provides a mathematical proof of its validity. The last section concludes the paper.

\section{THE functional graph OF THE GENERATED TENT MAP IN FIXED-POINT DOMAIN}
\label{sec:SMN}

A generalized Tent map, implemented in a fixed-point domain with fixed precision $e$, is defined on a discrete set. To simplify the representation of its functional graph, we scale all elements of the domain and range by a factor of $2^e$. Consequently, the discrete set of interest becomes $\{x\vvert x\in\mathcal{Z}_{2^e}\}$.
The resulting generalized Tent map, denoted as $T_e$, is given by
\begin{equation}\label{eq:define_F}
T_e(x)=
  \begin{cases}
    \mathrm{R}(2^eb+\frac{1-b}{a}x) & \mbox{if } 0\leq x< 2^ea; \\
    \mathrm{R}(\frac{2^{e}-x}{1-a}) & \mbox{if } 2^ea \leq x \leq 2^e, 
  \end{cases}
\end{equation}
where $a$ and $b$ are parameters satisfying $2^e a, 2^e b \in {0, 1, \ldots, 2^e}$, and $\mathrm{R}(\cdot)$ denotes an integer quantization function, such as floor, ceiling, or rounding.
The choice of quantization function $\mathrm{R}(\cdot)$ has only a minor effect on the quantized values and does not substantially alter the overall structure of the functional graph. Therefore, for simplicity, this paper considers only floor quantization.
To improve readability, frequently used notations are defined in Table~\ref{tab:list:symbol}.

\setlength{\tabcolsep}{3pt}
\begin{table}[!htb]
\centering
\caption{Nomenclature.}
\begin{tabular}{c|l}
    \hline
    {\bf Symbol} & {\bf Definition} \\ \hline
 $\mathcal{Z}_{2^e}$ & The set of integers    \\  \hline  
    $\mathbb{Z}$ & The set $\{0, 1, \ldots, 2^e-1\}$.   \\  \hline
   $k_{\rm l}$ & The derivative of the generalized Tent map when $0\leq x< 2^ea$.  \\ \hline
    $k_{\rm r}$  & The derivative of the generalized Tent map when $2^ea \leq x \leq 2^e$.  \\ \hline
\end{tabular}
\label{tab:list:symbol}
\end{table}

The functional graph of the generated Tent map~\eqref{eq:define_F} with precision $e$ is denoted as $\mathrm{T}_e$, which is built as follows: the $2^e+1$ possible states are viewed as $2^e+1$ nodes; a directed edge is constructed from node $x$ to node $y$ if $T_e(x)=y$. When $(a, b)=(11/2^4, 7/2^4)$, 
the functional graphs of $T_e(x)$ for$e=4, 5, 6$ are respectively illustrated in Fig.~\ref{fig:SMNf}.
Because the map is a piecewise linear function, its derivative is the main characteristic and the primary factor influencing its dynamical properties. In the functional graph, the effects of the derivative are primarily manifested in the in-degree of the nodes. This relationship is further detailed in Property~\ref{prop:dX}, which formalizes the relationship between the in-degree of the functional graph implemented by $T_e$ and the derivative of $T_e$.

\begin{figure}[!htb]
	\centering
	\begin{minipage}[t]{\ThreeImW}
		\centering
		\includegraphics[width=\ThreeImW]{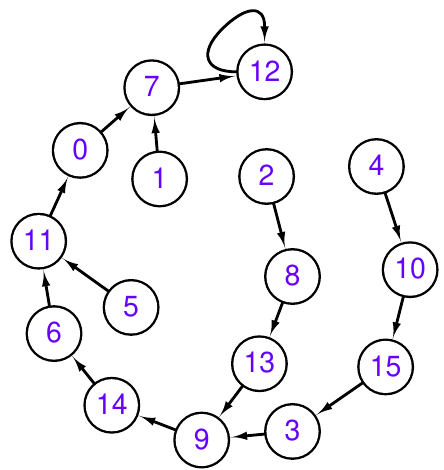}
		a)
	\end{minipage}
	\begin{minipage}[t]{\ThreeImW}
		\centering
		\includegraphics[width=\ThreeImW]{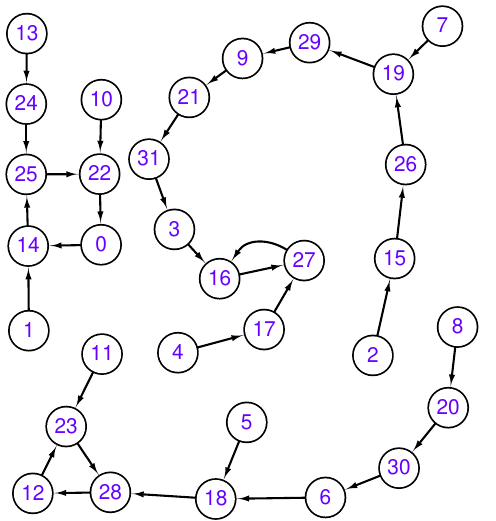}
		b)
	\end{minipage}
	\begin{minipage}[t]{0.75\OneImW}
		\centering
		\includegraphics[width=0.75\OneImW]{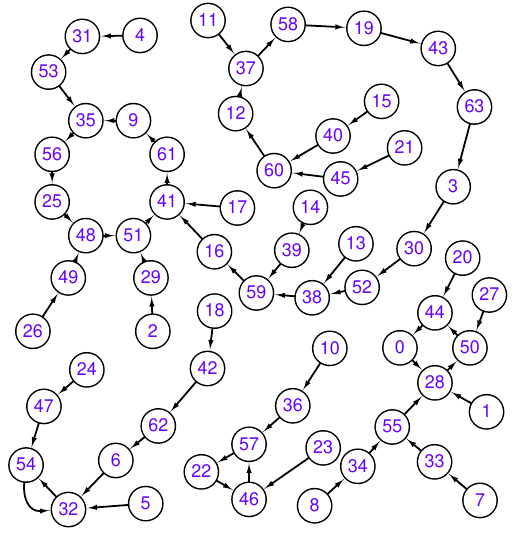}
		c)
	\end{minipage}
\caption{The structure of $\mathrm{T}_e$ with $(a, b)=(11/2^4, 7/2^4)$: 
a) $e=4$; b) $e=5$; c) $e=6$.}
\label{fig:SMNf}
\end{figure}

\begin{property}\label{prop:dX}
Let $d(y)$ denote the in-degree of node $y$ in $\mathrm{T}_e$. Then, one has
\begin{equation*}
d(y)\in
    \begin{cases}
    \{0, 1\}             & \mbox{if } y\in[0, 2^e b); \\
    \{\mathrm{R}(\frac{1}{k_{\rm l}}), \mathrm{R}(\frac{1}{k_{\rm l}})+1, \mathrm{R}(\frac{1}{k_{\rm l}})+2\} &  \mbox{if } y\in[2^e b, 2^e],      
    \end{cases}
\end{equation*}
where $k_{\rm l}=\frac{1-b}{a}$.
\end{property}
\begin{proof}
Let $\mathbf{S}(y)=\{x \vvert T_e(x)=y, x\in \mathcal{Z}_{2^e}\}$, the in-degree of node $y$ in $G_e$ is defined as 
$d(y)=|\mathbf{S}(y)|$. Since $\mathbf{S}(y)$ can be partitioned into two subsets based on whether the elements are in the interval $x\in[0, 2^ea)$ or $x\in[2^ea, 2^e]$, there is $d(y)=d_{\rm l}(y)+d_{\rm r}(y)$, where $d_{\rm l}(y)=|\{x\vvert x\in[0, 2^ea)\cap\mathbf{S}(y)\}|$ and $d_{\rm r}(y)= |\{x\vvert x\in[2^ea, 2^e]\cap\mathbf{S}(y)\}|$.
For $d_{\rm l}(y)$, two cases are divided by the range of $y$.
\begin{itemize}
    \item $y\in[0, 2^eb)$: It follows from Eq.~\eqref{eq:define_F} that there is no $x$ in $[0, 2^ea)$, hence $d_{\rm l}(y)=0$.
    
    \item $y\in[2^eb, 2^e)$: for any $x\in \{x\vvert x\in[0, 2^ea)\cap\mathbf{S}(y)\}$.
    Since Eq.~\eqref{eq:define_F} and the fact that $2^eb\in \mathcal{Z}_{2^e}$, one has $y=T_e(x)=2^eb+\mathrm{R}(k_{\rm l}x)$. This implying $\mathrm{R}(k_{\rm l}x)=y-2^eb$. Then $x \in [\frac{y-2^eb}{k_{\rm l}}, \frac{y-2^eb+1}{k_{\rm l}})\cap\mathbb{Z}$. So $d_{\rm l}(y)=|\{x \vvert x\in [\frac{y-2^eb}{k_{\rm l}}, \frac{y-2^eb+1}{k_{\rm l}})\cap\mathbb{Z}\}|\in\{\mathrm{R}(\frac{1}{k_{\rm l}}), \mathrm{R}(\frac{1}{k_{\rm l}})+1\}$.
\end{itemize}

For $d_\mathrm{r}(y)$, since Eq.~\eqref{eq:define_F}, any $x\in\{x\vvert x\in[2^ea, 2^e)\cap\mathbf{S}(y)\}$ satisfies $T_e(x)=\mathrm{R}\left(k_{\rm r}(x-2^e)\right)=y$.
Then $k_{\rm r}(x-2^e)\in[y, y+1)$, implying $x\in(\frac{y+1}{k_{\rm r}}+2^e, \frac{y}{k_{\rm r}}+2^e]\cap\mathbb{Z}$. Thus $d_{\rm r}=|\{x \vvert x\in (\frac{y+1}{k_{\rm r}}+2^e, \frac{y}{k_{\rm r}}+2^e]\cap\mathbb{Z}\}|\in\{\mathrm{R}(-\frac{1}{k_{\rm r}}), \mathrm{R}(-\frac{1}{k_{\rm r}}+1)\}$. Referring to $a\in[0, 1)$, one has $-\frac{1}{k_{\rm r}}\in(0,1)$, then $d_{\rm r}(y)\in\{0, 1\}$.
So according to $d(y)=d_{\rm l}(y)+d_{\rm r}(y)$, the property is proved.
\end{proof}

In the continuous domain, the generation Tent map has a unique and invariant point $c$ for any parameters of $a$ and $b$. This implying $f(c)=c$ with $c=\frac{1}{2-a}\in(a, 1]$. Referring to $f'(c)=-\frac{1}{1-a}<-1$, $c$ acts as a repeller making the interval $[a, f^{-1}(a))$ a repelling region, where $f^{-1}$ is the inverse function of $f$. It means that if $x\in[a, c)$, one has $c-x>c-f^2(x)$; if $x\in(c, f^{-1}(a)]$, one has $x-c>f^2(x)-c$, where \(f^i\) denotes the \(i\)-th iteration of $f$.
In the fixed-point domain, with a sufficiently large $e$, two intervals remain as repelling regions when $e$ is sufficiently large, as demonstrated by Property~\ref{prop:IP}.

\begin{property}\label{prop:IP}
When \( e \) is sufficiently large, $x$ in the intervals \([2^ea, 2^ec + \frac{k_{\rm r}}{k_{\rm r}^2-1}]\) or \((2^ec + \frac{1}{k_{\rm r}^2-1}, x']\) dictate that \( x > T^2_e(x) \) or \( x < T^2_e(x) \), respectively, where $x'=\max\{x\vvert x\geq 2^ea, T_e(x)\geq 2^ea\}$, \( k_{\rm r}=-\frac{1}{1-a} \) and \(c=\frac{1}{2-a}\).
\end{property}
\begin{proof}
When $x\in [2^ea, x']$, it follows from Eq.~\eqref{eq:define_F} that \begin{equation}\label{eq:T_e(x)in}
    T_e(x)\in(k_{\rm r}(x-2^e)-1, k_{\rm r}(x-2^e)]
\end{equation}
and $T^2_e(x)=\mathrm{R}(k_{\rm r}(T_e(x)-2^e))$, where $T^i_e(x)$ denotes the $i$-th iteration of $T_e$ applied to $x$. 
Since $k_{\rm r}<0$, one has $T^2_e(x)\in[k_{\rm r}^2(x-2^e)-2^ek_{\rm r}-1, k_{\rm r}^2(x-2^e)-2^ek_{\rm r}-k_{\rm r})$. 
Thus,
$$T^2_e(x)-x \in[\hat{x}-1, \hat{x}-k_{\rm r}),$$ 
where $\hat{x}=(k_{\rm r}^2-1)x-2^e(k_{\rm r}^2+k_{\rm r})$. Thus 
\begin{equation*}
    \begin{cases}
      x>T^2_e(x)  &\mbox{if }\hat{x}-k_{\rm r}\leq 0; \\
      x\leq T^2_e(x)        &\mbox{if }\hat{x}-1\geq 0.
    \end{cases}
\end{equation*}
Since $c=\frac{1}{2-a}$, it follows that $f(c)=c$ and $c=\frac{k_{\rm r}}{k_{\rm r}-1}$.
The condition $\hat{x}-k_{\rm r}\leq 0$ implies $x \leq\frac{2^e(k_{\rm r}^2+k_{\rm r})+k_{\rm r}}{k_{\rm r}^2-1}$.
Substitute $c$ to above inequation, one has $x\leq 2^ec+\frac{k_{\rm r}}{k_{\rm r}^2-1}$. Similarly, the condition $\hat{x}-1\geq 0$ implies $x\geq 2^ec+\frac{1}{k_{\rm r}^2-1}$.

The next step in the proof is to prove the existence of the two intervals $[2^ea, 2^ec+\frac{k_{\rm r}}{k_{\rm r}^2-1}]$ and $[2^ec+\frac{1}{k_{\rm r}^2-1}, x']$.
When $e$ is large enough, referring to $\frac{k_{\rm r}}{k_{\rm r}^2-1}$ is a constant and $c\in(a, 1]$, there is $2^ec+\frac{k_{\rm r}}{k_{\rm r}^2-1}>2^ea$.
Thus, $x\in[2^ea, 2^ec+\frac{k_{\rm r}}{k_{\rm r}^2-1}]$, one obtains $T^2_e(x)<x$.
Referring to the definition of $x'$, one deduces $T_e(x'+1)<2^ea$.
It follows from Eq.~\eqref{eq:T_e(x)in} that $k_{\rm r}(x'+1-2^e)< 2^ea$. Combining with $f(c)=c$, one obtains $k_{\rm r}(2^ec+2^e)=2^ec$, then $k_{\rm r}(x'+1-2^e)< k_{\rm r}(2^ec+2^e)+ 2^e(a-c)$
It follows from $k_{\rm r}<0$ that $x'+1>2^ec+k_{\rm r}2^e(a-c)$.
When $e$ is large enough, one has $k_{\rm r}2^e(a-c)-1$ approach $+\infty$, implying $x'>2^ec+\frac{1}{k_{\rm r}^2-1}$.
Thus, when $x\in[2^ec+\frac{1}{k_{\rm r}^2-1}, x']$, one obtains $T^2_e(x)>x$.
So, this property holds.
\end{proof}

Given the parameters $a=\frac{11}{16}, b=\frac{7}{16}$, one has determined that $c=\frac{1}{2-a}=\frac{16}{21}$ and $k_{\rm r}=-\frac{16}{5}$. For $e=4, 5, 6$, the intervals \([2^ea, 2^ec + \frac{k_{\rm r}}{k_{\rm r}^2-1}]\) are approximately $[11, 11.87]$, $[22, 24.05]$, $[44, 48.41]$, respectively.
These intervals, as shown in Fig.~\ref{fig:SMNf}, indicate that within them, the value of $x$ exceeds $T^2_e(x)$.
Additionally, the figure also illustrates the intervals \((2^ec + \frac{1}{k_{\rm r}^2-1}, x')\) where 
$x$ is less than \(T^2_e(x)\), providing a comprehensive view of the dynamic behavior of $x$ relative to $T^2_e(x)$ across different scales.

\section{The Optimization of the generalized Tent map in fixed-point domain}
\label{sec:Optimization}

The generalized Tent map with different parameters $a$ and $b$ exhibits distinct dynamical behaviors in the continuous domain. 
When $(a, b)\in\{(a, b) \vvert b<c, b<1-a\}\cup\{(a, b) \vvert b<c, b\geq 1-a, b>a\}$, the generalized Tent map exhibits chaos on the entire interval $[0, 1]$. This chaotic behavior makes the map well-suited for generating pseudo-chaotic sequences; however, it also presents several challenges, among which the short period is a key issue that needs to be addressed.
When quantizing and implementing a map in the digital domain, addressing the short period becomes crucial for constructing pseudorandom sequences. Adjusting the map's parameters can enhance the properties related to short-period, which is key to generating more effective and unpredictable pseudorandom sequences. This adjustment allows for better control over the sequence's characteristics, improving its suitability for applications in cryptography and other fields where chaotic dynamics are exploited.

Based on Property~\ref{prop:IP}, in contrast to the continuous domain, the path of nodes within the interval $(2^ec+\frac{k_{\rm r}}{k_{\rm r}^2-1}, 2^ec+\frac{1}{k_{\rm r}^2-1}]$ in the digital domain exhibits a clear dynamic degradation. This degradation is primarily due to the effects of the quantization function. Any $x$ in this interval may be a self-cycle or a node with a small period. As illustrated in Fig.~\ref{fig:SMNf}a) shown, when $e=4$, the interval is approximately $[11.87, 12.30)$ within this range, the value $12$ is a self-cycle. To avoid this problem, we can construct the generalized Tent map by the invariant point $c$. When $2^ec -\mathrm{R}(2^ec) =0.5$ with a sufficiently large $e$, interval $[2^ea, x')$ is a repelling region without short-period , as indicated by Corollary~\ref{coro:c0.5}.

\begin{corollary}\label{coro:c0.5}
When \( e \) is sufficiently large and $2^{e+1}c\in\mathbb{Z}$, the intervals \([2^ea, 2^ec )\) and \((2^ec, x')\) dictate that \( x > T^2_e(x) \) and \( x < T^2_e(x) \), respectively, where $x'=\max\{x\vvert x\geq 2^ea, T_e(x)\geq 2^ea\}$, \( k_{\rm r}=-\frac{1}{1-a} \), and \(c=\frac{1}{2-a}\).
\end{corollary}
\begin{proof}
It follows from Eq.~\eqref{eq:define_F} that $T_e(x)= \mathrm{R}(-k_{\rm r}(2^e-2^ec+2^ec-x))=\mathrm{R}(2^ec-k_{\rm r}(2^ec-x))$. Since $k_{\rm r}<-1$ and $2^{e+1}c\in\mathbb{Z}$, one can obtain 
    \begin{equation*}
        \begin{cases}
            T_e(x)>2^{e+1}c-x &\mbox{if } x\in[2^e a, 2^e c);\\
             T_e(x)<2^{e+1}c-x &\mbox{if } x\in(2^e c, x'].
        \end{cases}
    \end{equation*}
    Then, when $x\in[2^ea, 2^ec)$, one has $T_e(x)\in(2^{e+1}c-x, x')$, implying $T_e(x) \in(2^ec, x']$. Therefore, by substituting $x$ with $2^{e+1}c-x$ in above relation, one obtains $T^2_e(x)< x$.  
    Similarly, When $x\in(2^ec, x')$, one has $T^2_e(x)> x$.     
\end{proof}

It follows from Corollary~\ref{coro:c0.5} that a practical approach is constructing a generalized Tent map to avoid the short period by substituting the parameter $a$ with the invariant point $c$. Then, the generalized Tent map~\eqref{eq:define_F} can be transformed into 
\begin{equation*}
\label{eq:define_F2}
  T_e(x)=
  \begin{cases}
    \mathrm{R}(2^eb+k_{\rm l}{a}x)  & \mbox{if } 0\leq x< A;\\
    \mathrm{R}\left( k_{\rm r}(x-2^e)\right) & \mbox{if } A \leq x \leq 2^e, 
  \end{cases}
\end{equation*}
where $k_{\rm r}=-\frac{(2C+1)}{2^{e+1}-2C-1}$, $c=\mathrm{R}(2^{e}c)\in\mathcal{Z}_{2^e}$, and 
$A=\min\left\{x \vvert \mathrm{R}\left( k_{\rm r}(2^{e}-x\right)\leq 2^e\right\}$.

The functional graph $G_e$ comprises several connected components in the fixed-point domain. Any path initiated from a node will ultimately converge into a cycle within each connected component. These cycles within $G_e$ constitute the stable state space of the generalized Tent map. The duration of these cycles, or the period length, is directly related to the size of the stable state space, which means that a larger stable state space typically indicates a more extended period. Moreover, the extent of the stable state space is influenced by the in-degree of the nodes within the graph. Because any node in $G_e$ only points to a node, the larger average in-degree of a non-leaf node in $G_e$ means there is little stable state space.
Based on Property~\ref{prop:dX}, the in-degree of many nodes in $[0,2^eb]$ equals zero. we decrease the number of these nodes by adjusting the $2^eb$ by $k_{\rm l}$.  
Thus, the optimized generalized Tent map is defined as
\begin{equation}\label{eq:define_F2}
  \mathcal{T}_e(x)=
  \begin{cases}
    \mathrm{R}(k_{\rm l}(x+1))  & \mbox{if } 0\leq x< A;\\
    2^e                         & \mbox{if } x= A;\\
    \mathrm{R}\left(k_{\rm r}(x-2^e)\right) & \mbox{if } A < x \leq 2^e, 
  \end{cases}
\end{equation}
where $k_{\rm r}=-\frac{(2C+1)}{2^{e+1}-2C-1}$, 
\begin{equation}\label{eq:de:kl}
    k_{\rm l}= \frac{k_{\rm r}}{k_{\rm r}+1}, 
\end{equation}
$ A=\min\left\{x \vvert \mathrm{R}\left( k_{\rm r}(x-2^e\right)< 2^e\right\}-1$, and $c=\mathrm{R}(2^{e}c)\in\mathcal{Z}_{2^e}$. Property~\ref{prop:per} indicates that the stable state space of the map encompasses the entire state space, implying the map is a permutation.

\begin{property}\label{prop:per}
The optimized generalized Tent map~\eqref{eq:define_F2} is a permutation.
\end{property}
\begin{proof}
When $c$ is given, $k_{\rm r}, k_{\rm l}$ and $a$ are determined. Let ${\rm S_r}=\{\mathcal{T}_e(x) \vvert x\in(A, 2^e]\cap\mathbb{Z}\}$ and ${\rm S_l}=\Ze\setminus{\rm S_r}\setminus\{2^e\}$.
Thus, if $\mathcal{T}_{e, {\rm r}}:(A, 2^e]\cap\mathbb{Z} \rightarrow {\rm S_r}$ and $\mathcal{T}_{e, {\rm l}}:(0, A]\cap\mathbb{Z} \rightarrow {\rm S_l}$ are bijection, the map~\eqref{eq:define_F2} is a permutation.
Note that $\mathcal{T}_e(x)$ is decrement when $x\in(A, 2^e]\cap\mathbb{Z}$, namely $\mathcal{T}_{e, {\rm r}}$ is bijection. So if $\mathcal{T}_{e, {\rm l}}$ is bijection, this property holds.

Let the sequence $\{y_i\}$ consist of elements from the set ${\rm S_l}$ arranged in descending order. Because $|{\rm S_l}|=A$, one has $i\in\{1,2,\cdots, A\}$. There is 
\begin{equation}\label{eq:x'de}
x'_i=\min\{x \vvert x\in(A, 2^e), \mathcal{T}_e(x)<y_i\}.
\end{equation}
For any $i$, it follows from $y_i\in{\rm S_l}$ and $\mathcal{T}_e(A)=2^e$ that
\begin{equation}\label{eq:yiyi+1}
    y_i\in[\mathcal{T}_e(x'_i)+1,\mathcal{T}_e(x'_i-1)-1].
\end{equation}
Thus $-\frac{y_i}{k_{\rm r}}\in(2^e-x'_i,2^e -x'_i+1+\frac{1}{k_{\rm r}}]$. Since the definition of $k_{\rm l}$ and $k_{\rm r}$, one has $\frac{y_i}{k_{\rm l}}-\frac{y_i}{k_{\rm r}}=y_i$, then 
\begin{equation}\label{eq:frac}
\frac{y_i}{k_{\rm l}} \in \left(y_i-2^e+x'_i -1 -\frac{1}{k_{\rm r}}, y_i-2^e+x'_i \right].
\end{equation}

One can prove $\mathcal{T}_{e, {\rm l}}$ is bijection by constructing a map 
$ M': \{y_i\}\rightarrow \{x_i\}$, where $x_i=\max\{x \vvert \mathrm{R}(k_{\rm l}x_i)< y_i\}$. 
For any $i$, one has $y_i\in (k_{\rm l}x_i, k_{\rm l}(x_i+1)]$, implying 
$\frac{y_i}{k_{\rm l}}\in(x_i, x_i+1]$. Since Eq.~\eqref{eq:frac}, one deduces $x_i\in(y_i-2^e+x'_i -2-\frac{1}{k_{\rm r}}, y_i-2^e+x'_i )$. It follows from $x_i\in\Ze$ that 
\begin{equation}\label{eq:xi=}
x_i=y_i-2^e+x'_i-1.
\end{equation}
According to Eq.~\eqref{eq:frac}, one has $y_i\in(k_{\rm l}x_i-\frac{k_{\rm l}}{k_{\rm r}}, k_{\rm l}(x_i+1)]$, implying that $y_i>k_{\rm l}(x_i+1)-k_{\rm l}( 1+\frac{1}{k_{\rm r}})$. It follows from Eq.~\eqref{eq:de:kl} that $y_i>k_{\rm l}(x_i+1)-1$. Combining with the definition of $x_i$, one obtains $k_{\rm l}(x_i+1)\in[y_i, y_i+1)$, implying that $\mathrm{R}(k_{\rm l}(x_i+1))=y_i$.
So, there are two steps to prove $\mathcal{T}_{e, {\rm l}}$ is a bijection. First, it remains to prove that $M'$ is the inverse map of $\mathcal{T}_{e, {\rm l}}$.
For any $i$, to prove 
\begin{equation}\label{eq:xi+1i}
    x_{i+1}= x_i+1,
\end{equation}
there are two cases divided by $x'_i=x'_{i+1}$ or not:
\begin{itemize}
    \item $x'_i=x'_{i+1}$: since Eq.~\eqref{eq:xi=}, one deduces $y_{i+1}>y_{i}+1$. Referring to Eq.~\eqref{eq:x'de} and $x'_{i+1}=x'_i$, there are $\mathcal{T}_e(x'_i)<y_i$ and $\mathcal{T}_e(x'_{i}+1)>y_{i+1}$. Thus $\mathcal{T}_e(x'_i)<y_i+1<\mathcal{T}_e(x'_{i}+1)$, implying $y_i+1\in{\rm S_r}$. From $y_i<y_{i+1}$ and $y_i\in\mathbb{Z}$ for any $i$, one has 
    \begin{equation}\label{eq:yi+1yi}
        y_{i+1}=\min\{y\vvert y>y_i, y\in{\rm S_l}\}= y_i+1. 
    \end{equation}
    since Eq.~\eqref{eq:xi=}, one deduces $x_{i+1}= x_i+1$.
    
    \item $x'_i\neq x'_{i+1}$: since Eq.~\eqref{eq:x'de} and $y_i<y_{i+1}$, one has $y_{i+1}>\mathcal{T}_e(x'_{i+1})>y_i>\mathcal{T}_e(x'_i)$. First,
One can then prove that 
\begin{equation}\label{eq:induTex'}
    \mathcal{T}_e(x'_{i}-z)=y_i+z
\end{equation}
where $z\in[1, y_{i+1}-y_i-1]\cap\mathbb{Z}$ by induction on $z$.
For $z=1$, since Eq.~\eqref{eq:x'de}, one deduces $\mathcal{T}_e(x'_{i+1}-1)\geq y_i$. Because $y_i\not\in{\rm S_r}$, one has $\mathcal{T}_e(x'_{i}-1)\geq y_i+1$. Assume $\mathcal{T}_e(x'_{i}-1)> y_i+1$, one obtains $y_i+1\not\in{\rm S_r}$. It follows from $y_i+1<2^e$ that $y_i+1\in {\rm S_l}$, this implying $y_{i+1}=y_i+1$. Then $x’_{i+1}=x'_i$, which contradict with $x’_{i+1}\neq x'_i$. Thus $\mathcal{T}_e(x'_{i}-1)= y_i+1$. 
This implies the statement holds for $z=1$. Assume $\mathcal{T}_e(x'_{i}-z)=y_i+z$ for $z=n<y_{i+1}-y_i-1$. 
Since $k_{\rm r}<-1$, one deduces $\mathcal{T}_e(x'_{i}-n-1)\geq \mathcal{T}_e(x'_{i}-n)+1$, this implying $\mathcal{T}_e(x'_{i}-n-1)\geq y_i+n+1$. Similarly to the case of $z=1$, assume $\mathcal{T}_e(x'_{i}-n-1)> y_i+n+1$, one obtains $y_{i+1}=y_i+n+1$, which contradict with $n<y_{i+1}-y_i-1$. Thus $\mathcal{T}_e(x'_{i}-n-1)= y_i+n+1$.
This implies the statement holds for $i=n+1$.
Therefore, this completes the induction of Eq.~\eqref{eq:induTex'}. Next, referring to Eq.~\eqref{eq:induTex'}, one has $\mathcal{T}_e(x'_{i}-y_{i+1}+y_i+1)=y_{i+1}-1$. Since Eq.~\eqref{eq:x'de}, one obtains $x'_{i+1}=x'_{i}-y_{i+1}+y_i+1$. Combining with Eq.~\eqref{eq:xi=}, $x_{i+1}=x'_{i}+y_i-2^e=x_i+1$.
\end{itemize}
Since the definition of $y_i$ and $\mathcal{T}_e(2^e)=0$, one deduces $y_1=\mathcal{T}_e(x'_1)+1 > 0$. Assume $\mathcal{T}_e(x'_1)> 2^e-x'_1$, there is $y\not\in{\rm S_r}$ such that $y\in[0, 2^e-x'_1]$, which contradicts with $y_1>\mathcal{T}_e(x'_1)$. 
Thus $\mathcal{T}_e(x'_1)\leq 2^e-x'_1$. Since $k_{\rm r}<-1$, one has $\mathcal{T}_e(x'_1)=\mathrm{R}(k_{\rm r}(x'_1-2^e))\geq2^e-x'_1$, then $\mathcal{T}_e(x'_1)=2^e-x'_1$. Therefore, $y_1=2^e-x'_1+1$. It follows from Eq.~\eqref{eq:xi=} that $x_1=0$. So from Eq.~\eqref{eq:xi+1i} and $i\in\{1,2,\cdots,A\}$, one deduces $x_i=i-1$. Thus, $M'$ is the inverse map of $\mathcal{T}_{e, {\rm l}}$.
Then, to prove that $M$ is a bijection, one needs to establish both injectivity and surjectivity.
To prove $M$ is injective, assume $M(y_j)=M(y_k)$, namely $x_j=x_k$.
Since Eq.~\eqref{eq:xi+1i}, one has $j=k$, implying that $y_j=y_k$. So
$M$ is injective. Next, we demonstrate that $M$ is surjective. According to the definition of $x_i$, for any $x_i$, there is $y_i$ such that $M(y_i)=x_i$. This proves that $M$ is surjective.
Since $M$ is both injective and surjective, it follows that $M$ is a bijection. So $\mathcal{T}_{e, {\rm l}}$ is a bijection. This property holds.
\end{proof} 

\begin{figure}[!htb]
	\centering
	\begin{minipage}[t]{1\twofigwidth}
		\centering
		\includegraphics[width=1\twofigwidth]{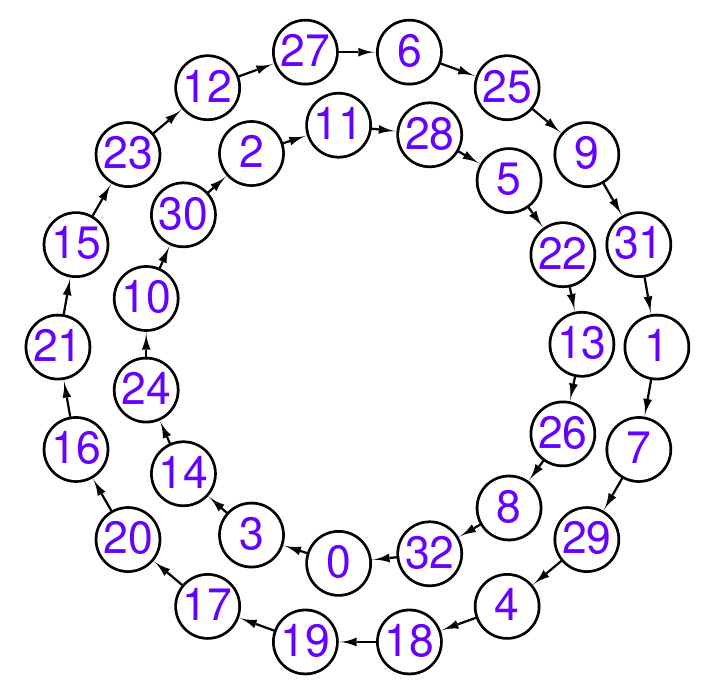}
            a)
	\end{minipage}
        \begin{minipage}[t]{1\twofigwidth}
	\centering
		\includegraphics[width=1\twofigwidth]{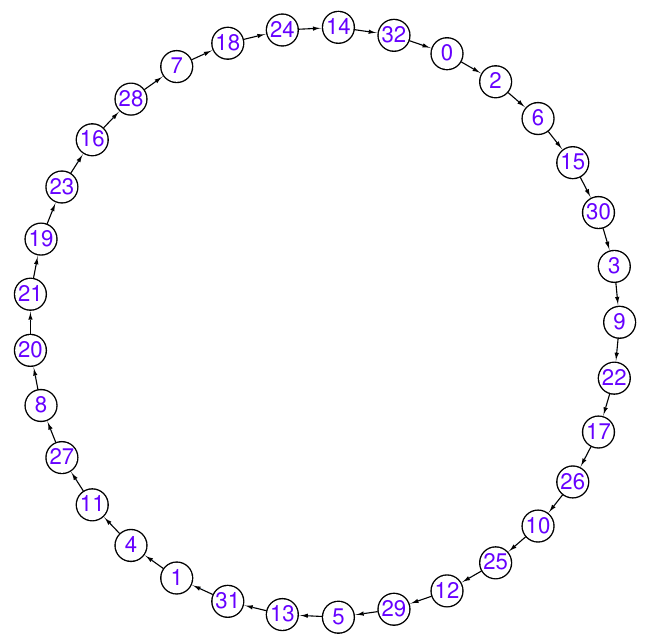}
            b)
	\end{minipage}
\caption{The structure of $\mathrm{T}_e$ with $c=18$: a) $e=5$; b) $e=20$.}
\label{fig:LE}
\end{figure}

The Lyapunov exponent characterizes the average exponential divergence rate of adjacent trajectories in phase space and is one of the numerical characteristics for identifying chaotic motion. The Lyapunov exponent of maps defined in discrete sets as the fixed-pointed domain can describe the sensitivity to initial conditions or the stability of dynamical systems. The Lyapunov exponent of the optimized mapping~\eqref{eq:define_F2} with initial value $x_0$ defined as 
 \begin{equation*}
     \lambda_{(\mathcal{T},x_0)}=\frac{1}{l_x} \sum_{n=0}^{l_x-1} \ln |\mathcal{T}(x_n + 1)-\mathcal{T}(x_n)|,
 \end{equation*}
which is proposed in \cite{Kocarev:Discrete:2006}. The discrete Lyapunov exponent of $\mathcal{T}_e(x)$ is defined as
 \begin{equation*}
     \lambda_{\mathcal{T}}=\frac{1}{2^e} \sum_{x_0=0}^{2^e}\lambda_{(\mathcal{T}, x_0)},
 \end{equation*}
The variation of the Lyapunov exponents $\lambda_{\mathcal{T}}$ for $e=16$ versus $c$ are illustrated in Fig.~\ref{fig:LE}.

\begin{figure}[!htb]
	\centering
	\begin{minipage}[t]{1\OneImW}
		\centering
			\includegraphics[width=1\OneImW]{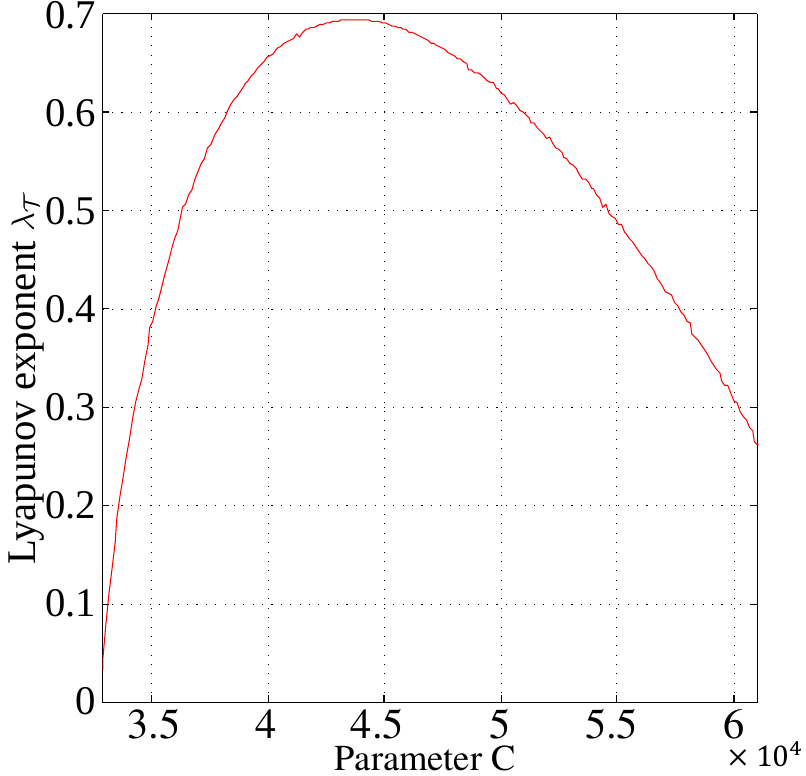}
	\end{minipage}
\caption{Dependence of Lyapunov exponent $\lambda_{\mathcal{T}}$ on parameter $c$.}
\label{fig:LE}
\end{figure}


\section{CONCLUSION}

This paper presents a comprehensive theoretical analysis of the graph structure of generalized Tent maps in the fixed-point domain. 
the graph structure is studied for its dynamic properties at indeterminate points.
It then proposes an optimized generalized Tent map, which is proven to have no short period and acts as a permutation.
It confirmed the study that implemented maps in digital devices need analysis, or the result may differ.

\bibliographystyle{IEEEtran_doi}
\bibliography{GTent}

\end{document}